\documentclass[envcountsame]{llncs}

\usepackage{makeidx}  
\usepackage{amsmath, amssymb, amsfonts}
\usepackage{mathtools}
\usepackage{mathrsfs}
\usepackage[shortlabels]{enumitem}
\usepackage{calrsfs}

\usepackage{tikz}

\usetikzlibrary{arrows,decorations.pathmorphing,backgrounds,positioning,fit,automata}

\usetikzlibrary{decorations.markings}
\usetikzlibrary{decorations.pathreplacing}
\usepackage{authblk}
\usepackage{pgfplots}

\usepackage[colorlinks=true,breaklinks=true,bookmarks=true,urlcolor=magenta,
     citecolor=blue,linkcolor=red,bookmarksopen=false,draft=false]{hyperref}

\tikzset{mylabel/.style  args={at #1 #2  with #3}{
    postaction={decorate,
    decoration={
      markings,
      mark= at position #1
      with  \node [#2] {#3};
 } } } }

\newenvironment{customlegend}[1][]{%
    \begingroup
    \csname pgfplots@init@cleared@structures\endcsname
    \pgfplotsset{#1}%
}{%
    \csname pgfplots@createlegend\endcsname
    \endgroup
}%

\def\addlegendimage{\csname pgfplots@addlegendimage\endcsname}

\newcommand{\eqflow}[1]{#1^{*}}
\newcommand{\optflow}[1]{\widetilde{#1}}
\newcommand{\expon}[1]{\mathrm{e}^{#1}}
\newcommand{\Opt}{\operatorname{\mathsf{Opt}}}
\newcommand{\PoA}{\operatorname{\mathsf{PoA}}}

\newcommand{\WEq}{\operatorname{\mathsf{WEq}}}
\newcommand{\diff}{\ \mathrm{d}}

\spnewtheorem{newclaim}[theorem]{Claim}{\bfseries}{\itshape}

\begin{document}

\title{On the Price of Anarchy of Highly Congested Nonatomic Network Games}
\titlerunning{Limit of the Price of Anarchy}

\author{Riccardo Colini-Baldeschi\inst{1}
\and Roberto Cominetti\inst{2}
\and Marco Scarsini\inst{1}}

\institute{Dipartimento di Economia e Finanza, LUISS, Viale Romania 32, 00197 Roma, Italy, \email{rcolini@luiss.it}, \email{marco.scarsini@luiss.it}
\and 
Facultad de Ingenier\'ia y Ciencias, Universidad Adolfo Ib\'a\~nez, Santiago, Chile, \email{roberto.cominetti@uai.cl}
}

\maketitle

\begin{abstract}
We consider nonatomic network games with one source and one destination. We examine the asymptotic behavior of the price of anarchy as the inflow increases. In accordance with some empirical observations, we show that, under suitable conditions, the price of anarchy is asymptotic to one. We show with some counterexamples that this is not always the case. The counterexamples occur in very simple parallel graphs.
\end{abstract}

\section{Introduction}\label{se:intro}

The analysis of network routing costs and their efficiency goes back at least to Pigou \cite{Pig:Macmillan1920}, who, in the first edition of his book introduces his famous two-road model. 
Wardrop \cite{War:PICE1952} develops a model where many players (vehicles on the road) choose a road in order to minimize their cost (traveling time) and the influence of each one of them, singularly taken, is negligible. He introduces a concept of equilibrium that has become the standard in the literature on nonatomic network games. 

When travelers minimize their traveling time without considering the negative externalities that their behavior has on other travelers, the collective outcome of the choices of all travelers is typically inefficient, i.e., it is worse than the outcome that a benevolent planner would have achieved. Various measures have been proposed to quantify this inefficiency. Among them the price of anarchy has been the most successful. Introduced by Koutsoupias and Papadimitriou
\cite{KouPap:STACS1999} and given this name by Papadimitriou
\cite{Pap:PACM2001}, it is the ratio of the worst social equilibrium cost and the optimum cost.
The price of anarchy has been studied by several authors and interesting bounds for it have been found under some conditions on the cost functions.

Most of the existing results about the price of anarchy consider worst-case  scenarios. These results are not necessarily helpful in specific situations. 
In a nice recent paper O'Hare et al. \cite{OHaConWat:TRB2016} show, both theoretically and with the aid of simulations, how the price of anarchy is affected by changes in the total inflow of players. They consider data for three cities and they  write 

\emph{In each city, it can be seen that there are broadly three identifiably distinct regions of behaviour: an initial region in which the Price of Anarchy is one; an intermediate region of fluctuations; and a final region of decay, which has a similar characteristic shape across all three networks. The similarities in this general behaviour across the three cities suggest that there may be common mechanisms that drive this variation.
}

The core of the paper  by \cite{OHaConWat:TRB2016} is an analysis of the intermediate fluctuations. In our paper we will mainly look at  the asymptotic behavior of the price of anarchy. 
We consider nonatomic congestion games with single source and single destination. We show that for a large class of cost functions the price of anarchy is, indeed, asymptotic to one, as the mass of players grows. Nevertheless, we can find counterexamples where its $\limsup$ is not $1$ and it can even be infinite.

\subsubsection*{Contribution}

The goal of this paper is twofold. On one hand we provide some positive results that show that under some conditions the price of anarchy of nonatomic network games is indeed asymptotic to one. On the other hand, we present some counterexamples where the $\limsup$ of the price of anarchy is not one.

In particular, first we show that, for any single-source, single-destination graph, the price of anarchy is asymptotic to one whenever the cost of at least one path is bounded. Then we move the analysis to parallel graphs and   we show that in this class the price of anarchy is asymptotic to one for a large class of cost functions that we characterize in terms of regularly varying functions (see \cite{BinGolTeu:CUP1989} for properties of these functions). This class of cost functions includes affine functions and cost functions that can be bounded by a pair of affine functions with the same slope.
 
Counterexamples can be found where the behavior of the price of anarchy is periodic on a logarithmic scale, therefore its $\limsup$ is larger than one both as the mass of players grows unbounded and as it goes to zero. In another counterexample the $\limsup$ of the price of anarchy is infinite. 
A further counterexample shows that the price of anarchy may not converge to one even for convex cost functions. What is interesting is that all the counterexamples concern a very simple parallel graph with just two edges. Therefore the bad behavior of the price of anarchy depends solely on the costs and not on the topology of the graph. This is in stark contrast with the results in \cite{OHaConWat:TRB2016}, where the irregular behavior of the price of anarchy in the intermediate region of inflow heavily depends on the structure of the graph.

\subsubsection*{Related literature}

Wardrop's nonatomic model has been studied by Beckmann et al. \cite{BecMcGWin:Yale1956} and many others. 
The formal foundation of games with a continuum of players came with Schmeidler \cite{Sch:JSP1973} and then with Mas Colell \cite{Mas:JME1984}.
Nonatomic congestion  games have been studied, among others, by Milchtaich 
\cite{Mil:MOR2000,Mil:JET2004}.

Various bounds for the price of anarchy in nonatomic games have been proved, under different conditions.
In particular Roughgarden and Tardos \cite{RouTar:JACM202} prove that, when the cost functions are affine, the price of anarchy in nonatomic games is at most $4/3$, irrespective of the topology of the network. The bound is sharp and is attained even in very simple networks.  Several authors have extended this bound to larger classes of functions. 
Roughgarden \cite{Rou:JCSS2003} shows that if the class of cost functions includes the constants, then the worst price of anarchy is achieved on parallel networks with just two edges. In his paper he considers bounds for the price of anarchy when the cost functions are polynomials whose degree is at most $d$.
Dumrauf and Gairing \cite{DumGai:INE2006} do the same when the degrees of the polynomials are all between $s$ and $d$.
 Roughgarden and Tardos \cite{RouTar:GEB2004} provide a unifying result for the class of standard costs, i.e., costs $c$ that are differentiable and such that $x c(x)$ is convex. 
Correa et al. \cite{CorSchSti:MOR2004} consider the price of anarchy for networks where edges have a capacity and costs are not necessarily convex, differentiable, or even continuous. In \cite{CorSchSti:GEB2008} they reinterpret and extend these results using a geometric approach.
In \cite{CorSchSti:OR2007} they consider the problem of minimizing the maximum latency rather than the average latency and provide results about the price of anarchy in this framework. The reader is referred to \cite{RouTar:AGT2007,Rou:AGT2007} for a survey of the literature.

Some papers show how in real life the price of anarchy may substantially differ from the worst-case scenario, \cite{YouGasJeo:PRL2008,
LawHuaLiu:IEEETWC2012}. The papers
Gonz\'alez Vay\'a et al.
\cite{GonGraAndLyg:IEEECDC2015} deal with a problem of optimal schedule
for the electricity demand of a fleet of plug-in electric vehicles. Without using the term, they show that the price of anarchy goes to one as the number of vehicles grows.
Cole and Tao 
\cite{ColTao:ArXiv2015} study large Walrasian auctions and large Fisher markets and show that in both cases  the price of anarchy goes to one as the market size increases.
Feldman et al. \cite{FelImmLucRouSyr:ArXiv2015} define a concept of $(\lambda,\mu)$-smoothness for sequences of games, and show that the price of anarchy in atomic congestion games converges to the price of anarchy of the corresponding nonatomic game, when the number of players grows. 
\cite{Pat:TS2004} and \cite{JosPat:TRB2007} perform sensitivity analysis of Wardrop equilibrium to some parameters of the model.
Closer to the scope of our paper, Englert et al. \cite{EngFraOlb:TCS2010} examine how the equilibrium of a congestion game changes when either the total mass of players is increased by $\varepsilon$ or an edge that carries an $\varepsilon$ fraction of the mass is removed. For polynomial cost functions they bound the increase of the equilibrium cost when a mass $\varepsilon$ of players is added to the system.

\section{The model}\label{se:model}

Consider a finite directed multigraph $\mathcal{G}=(V,E)$, where $V$ is a set of vertices and $E$ is a set of edges. The graph $G$ together with a source $s$ and a destination $t$ is called a network.  A path $P$ is a set of consecutive edges that go from source to destination. Call $\mathcal{P}$ the set of all paths. Each path $P$ has a flow $x_{P}\ge 0$ and call  $\boldsymbol{x}=(x_{P})_{P\in\mathcal{P}}$. The total flow from source to destination is denoted by $M\in\mathbb{R}_{+}$. A flow $\boldsymbol{x}$ is \emph{feasible} if
$\sum_{P\in\mathcal{P}}x_{P}=M$.
Call $\mathcal{F}_{M}$ the set of feasible flows.
For each edge $e\in E$ there exists a cost function $c_{e}(\cdot):\mathbb{R}_{+}\to\mathbb{R}_{+}$, that is assumed (weakly) increasing and continuous. Call $\boldsymbol{c}=(c_{e})_{e\in E}$.
This defines a \emph{nonatomic congestion game} $\Gamma_{M}=(\mathcal{G},M,\boldsymbol{c})$. The number $M$ can be seen as the mass of players who play the game. 

The cost of a path $P$ with respect to a flow $\boldsymbol{x}$ is the sum of the cost of its edges: 
$c_{P}(\boldsymbol{x})=\sum_{e\in P}c_{e}(x_{e})$, where 
\[
x_{e}=\sum_{\substack{P\in\mathcal{P}:\\e\in P}}x_{P}.
\]

For each flow $\boldsymbol{x}$ define the \emph{social cost} associated to it as
\[
C(\boldsymbol{x}) := \sum_{P\in\mathcal{P}}x_{P}c_{P}(\boldsymbol{x})=\sum_{e\in E} x_e c_e(x_e).
\]

A flow $\eqflow{\boldsymbol{x}}$ is an \emph{equilibrium flow} if for every $P,Q\in\mathcal{P}$ such that $\eqflow{x}_{P}>0$ we have
$c_{P}(\eqflow{\boldsymbol{x}}) \le c_{Q}(\eqflow{\boldsymbol{x}}).$
Call $\mathcal{E}(\Gamma_{M})$ the set of equilibrium flows in $\Gamma_{M}$ and define
$\WEq(\Gamma_{M}) = \max_{\boldsymbol{x}\in\mathcal{E}(\Gamma_{M})}C(\boldsymbol{x})$
the \emph{worst equilibrium cost} of $\Gamma_{M}$.  Actually, in the present setting the cost $C(\eqflow{\boldsymbol{x}})$ is the same for every equilibrium $\eqflow{\boldsymbol{x}}$ 
(see \cite{FloHea:HTS2003}).

A flow $\optflow{\boldsymbol{x}}$ is an \emph{optimum flow} if 
$C(\optflow{\boldsymbol{x}}) = \min_{\boldsymbol{x}\in\mathcal{F}_{M}}C(\boldsymbol{x}).$
Call $\mathcal{O}(\Gamma_{M})$ the set of optimum flows in $\Gamma_{M}$ and define
$\Opt(\Gamma_{M})=C(\optflow{\boldsymbol{x}}), \quad\text{for }\optflow{\boldsymbol{x}} \in \mathcal{O}(\Gamma_{M})$
the \emph{optimum cost} of $\Gamma_{M}$.

The \emph{price of anarchy} of the game $\Gamma_{M}$ is defined as
\begin{equation*}
\PoA(\Gamma_{M}):=\frac{\WEq(\Gamma_{M})}{\Opt(\Gamma_{M})}.
\end{equation*}

We will be interested in the price of anarchy of this game, as $M\to\infty$. We will show that, under some conditions, it is asymptotic to one. We call \emph{asymptotically well behaved} the congestion games for which this happens.

\section{Well behaved congestion games}
\label{se:well_behaved}

\subsection{General result}\label{suse:general}

The following general result shows that for any network the price of anarchy is asymptotic to one when at least one path has a bounded cost.

\begin{theorem}\label{th:ConstantCost}
For each path $P\in\mathcal{P}$  denote 
\begin{equation*}
c_P^\infty=\sum_{e\in P}c_e^{\infty}\quad\text{with}\quad c_e^\infty=\lim_{z\to\infty}c_e(z)
\end{equation*} 
and suppose that $B:=\min_{P\in\mathcal{P}}c_P^\infty$
is finite. Then, $\lim_{M\to\infty} \PoA(\Gamma_M) = 1$.
\end{theorem}

\begin{proof}
Let $\eqflow{\boldsymbol{x}}$ be an equilibrium for $\Gamma_M$. Then if $\eqflow{x}_{P}>0$ we have
\begin{equation*}
c_P(\eqflow{\boldsymbol{x}})=\min_{Q\in\mathcal{P}}c_Q(\eqflow{\boldsymbol{x}})\leq \min_{Q\in\mathcal{P}}c_Q^\infty = B
\end{equation*}
and therefore 
\begin{equation*}
\WEq(\Gamma_M)=\sum_{P\in\mathcal{P}}\eqflow{x}_P c_P(\eqflow{\boldsymbol{x}})\leq \sum_{P\in\mathcal{P}}\eqflow{x}_PB=MB.
\end{equation*}
It follows that 
\begin{equation*}
\PoA(\Gamma_M)\leq \frac{MB}{\Opt(\Gamma_M)},
\end{equation*}
so that it suffices to prove that $\Opt(\Gamma_M)/M\to B$.
To this end denote $\Delta(\mathcal{P})$ the simplex defined by
$\boldsymbol{y}=(y_P)_{P\in\mathcal{P}}\geq 0$ and $\sum_{P\in\mathcal{P}}y_P=1$, 
so that
\begin{align*}
\frac{1}{M}\Opt(\Gamma_M)&=\min_{\boldsymbol{x}\in\mathcal{F}_M}\sum_{P\in\mathcal{P}}\frac{x_P}{M}c_P(\boldsymbol{x})\\
&=\min_{\boldsymbol{y}\in\Delta(\mathcal{P})}\sum_{P\in\mathcal{P}}y_Pc_P(M\boldsymbol{y}).
\end{align*}
Denote $\Phi_M(\boldsymbol{y})=\sum_{P\in\mathcal{P}}y_Pc_P(M\boldsymbol{y})$. Since the cost functions
$c_e(\cdot)$ are non-decreasing, the family $\Phi_M(\cdot)$ monotonically increases with $M$ 
towards the limit function
\begin{equation*}
\Phi_\infty(y)=\sum_{P\in\mathcal{P}:y_P>0}y_Pc_P^\infty.
\end{equation*}
Now we use the fact that a monotonically increasing  family of functions epi-converges (see \cite{Att:Pitman1984}) and since $\Delta(\mathcal{P})$ is
compact it follows that the minimum $\min_{\boldsymbol{y}\in\Delta(\mathcal{P})}\Phi_M(\boldsymbol{y})$
converges as $M\to\infty$ towards
\begin{equation*}
\min_{\boldsymbol{y}\in\Delta(\mathcal{P})}\Phi_\infty(\boldsymbol{y}).
\end{equation*}
Clearly this latter optimal value is $B$ and is attained by setting $y_P>0$ only on those paths
$P$ that attain the smallest value $c_P^\infty=B$, and therefore we conclude 
\begin{equation*}
\frac{1}{M}\Opt(\Gamma_M)=\min_{y\in\Delta(\mathcal{P})}\Phi_M(\boldsymbol{y})\to B,
\end{equation*}
as was to be proved. \qed
\end{proof}

\subsection{Parallel graphs}
\label{suse:parallel}

In this section we examine the asymptotic behavior of the price of anarchy when the game is played on a parallel graph. 

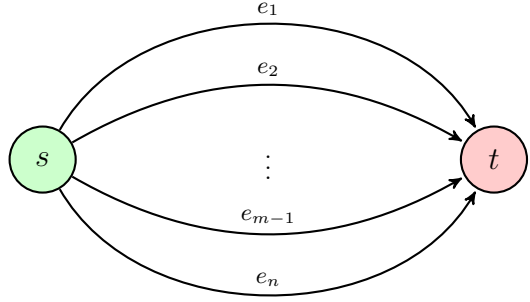
\begin{figure}[h]
\centering
\begin{tikzpicture}[->,>=stealth',shorten >=1pt,auto,node distance=3cm,
  thick,main node/.style={circle,fill=blue!20,draw,minimum size=25pt,font=\sffamily\large\bfseries},source node/.style={circle,fill=green!20,draw,minimum size=25pt,font=\sffamily\large\bfseries},dest node/.style={circle,fill=red!20,draw,minimum size=25pt,font=\sffamily\large\bfseries}]

  \node[source node] (1) {$s$};
  \node (m) [right of=1] {$\vdots$};
  \node[dest node] (2) [right of=m] {$t$};

  \path[every node/.style={font=\sffamily\small}]
    (1) edge [bend left = 60] node[above] {$e_{1}$} (2)
        edge [bend left = 30] node[above] {$e_{2}$} (2)
        edge [bend right = 30] node[above] {$e_{m-1}$} (2)
        edge [bend right = 60] node[above] {$e_{n}$} (2);

\end{tikzpicture}
~\vspace{0cm} \caption{\label{fi:parallelnetwork} General parallel network.}
\end{figure}

Let $\mathcal{G} = (V,E)$ be a parallel graph such that $V = \{s,t\}$ are the vertices and $E = \{e_1,e_2,\ldots,e_{n}\}$ are the edges, as in Figure~\ref{fi:parallelnetwork}.
For each edge $e_{i}\in E$ the function $c_i(\cdot)$ represents the cost function of the edge $e_{i}$. Call $\Gamma_{M}=(\mathcal{G},M,\boldsymbol{c})$ the corresponding game. In the whole section we will deal with this graph.
%

\subsubsection{Adding a constant to costs}

First we prove a preservation result. We show that if the price of anarchy of a game converges to $1$, then adding positive constants to each cost does not alter this asymptotic behavior.

\begin{theorem}\label{th:addconstant}
Given a game  $\Gamma_{M}=(\mathcal{G},M,\boldsymbol{c})$ and a  vector $\boldsymbol{a}\in [0,\infty)^{n}$, consider a new game $\Gamma_{M}^{\boldsymbol{a}}(\mathcal{G},M,\boldsymbol{c}^{\boldsymbol{a}})$, where
\begin{equation*}
c_{i}^{\boldsymbol{a}}(x) = a_{i}+c_{i}(x).
\end{equation*}
If $c_i(\cdot)$ is strictly increasing and continuous, $\lim_{x\to\infty}c_{i}(x)=\infty$ for all $e_{i}\in E$, and\, $\lim_{M \to \infty} \PoA(\Gamma_M) = 1$, then $\lim_{M \to \infty} \PoA(\Gamma_{M}^{\boldsymbol{a}}) = 1$.
\end{theorem}

\subsubsection{Regularly varying functions}

\begin{definition}
Let $\beta \ge 0$.
A function $\Theta : (0, +\infty) \to (0, +\infty)$ is called \emph{$\beta$-regularly varying} if for all $a>0$
\begin{equation*}
\lim_{x \to \infty} \frac{\Theta(a\cdot x)}{\Theta(x)} = a^{\beta} \in (0,+\infty).
\end{equation*}
When $\beta=1$, we just say that the function is regularly varying.
\end{definition}

The following theorem shows that asymptotically the price of anarchy goes to $1$ for a large class of cost functions.

\begin{theorem}\label{th:regvar1}
Consider the game $\Gamma_{M}$ and suppose that for some $\beta>0$ there exists a $\beta$-regularly varying function $c(\cdot) \in C^{1}$ such that the function $x\mapsto c(x)+xc'(x)$ is strictly increasing and for all $e_{i}\in E$ the function $c_i(\cdot)$ is strictly increasing and continuous with 
\begin{equation}\label{eq:regvar1}
\lim_{x \to \infty} \frac{c^{-1} \circ c_i(x)}{x} = \alpha_i \in (0,+\infty]
\end{equation}
and that at least one $\alpha_{i}$ is finite.
Then 
\begin{equation*}
\lim_{M \to \infty} \PoA(\Gamma_{M}) = 1.
\end{equation*}
\end{theorem}

\begin{proof}
We begin by noting that if some cost $c_i(\cdot)$ is bounded, then the result follows directly from Theorem~\ref{th:ConstantCost}. 
Suppose now that $c_i(x)\to\infty$ when $x\to\infty$
in all links and consider first the case where all the $\alpha_i$ are finite. In this case the equilibrium flows $\eqflow{x}_i$ must diverge to $\infty$ as $M\to\infty$ and the 
equilibrium is characterized by $c_i(\eqflow{x}_i)=\lambda$. This allows to derive an upper bound for the 
cost of the equilibrium. That is, \eqref{eq:regvar1} implies that for small $\varepsilon>0$ we have 
\begin{equation*}
\frac{c^{-1}\circ c(\eqflow{x}_i)}{\eqflow{x}_i}=\frac{c^{-1}(\lambda)}{\eqflow{x}_i}\in(\alpha_i-\varepsilon,\alpha_i+\varepsilon),
\end{equation*}
provided $M$ is large enough. It then follows that
\begin{equation*}
\sum_{i=1}^n\frac{c^{-1}(\lambda)}{\alpha_i+\varepsilon}\leq\sum_{i=1}^n\eqflow{x}_i=M,
\end{equation*}
so that, denoting 
\begin{equation*}
a(\varepsilon)=\left(\sum_{i=1}^n\frac{1}{\alpha_i+\varepsilon}\right)^{-1},
\end{equation*} 
we get 
$\lambda\leq c(Ma(\varepsilon))$ and 
\begin{equation*}
\WEq=M\lambda\leq M c(Ma(\varepsilon)).
\end{equation*}

Next we derive a lower bound for the optimal cost
\begin{equation*}
\Opt(\Gamma_M)=\min_{x\in\mathcal{F}_M}\sum_{i=1}^n x_ic_i(x_i).
\end{equation*}
We note that when $M\to\infty$ the optimal solutions are such that $x_i(M)\to\infty$ so that using \eqref{eq:regvar1} and the fact that $\alpha_i-\varepsilon>0$ we get for all $M$ large enough
\begin{equation*}
\min_{x\in\mathcal{F}_M}\sum_{i=1}^n x_ic_i(x_i)\geq \min_{x\in\mathcal{F}_M}\sum_{i=1}^n x_ic((\alpha_i-\varepsilon)x_i).
\end{equation*}
The optimality condition for the latter yields
\begin{equation*}
c((\alpha_i-\varepsilon)x_i)+(\alpha_i-\varepsilon)x_i c'((\alpha_i-\varepsilon)x_i)=\mu.
\end{equation*}
For the sake of brevity we denote $\tilde c(x)=c(x)+xc'(x)$ and $y_i=(\alpha_i-\varepsilon)x_i$
so that the optimality condition becomes $\tilde c(y_i)=\mu$. This yields $y_i=\tilde c^{-1}(\mu)$
and therefore
\begin{equation*}
M=\sum_{i=1}^n x_i=\sum_{i=1}^n\frac{\tilde c^{-1}(\mu)}{\alpha_i-\varepsilon}.
\end{equation*}
Denoting 
\begin{equation*}
b(\varepsilon)=\left(\sum_{i=1}^n\frac{1}{\alpha_i-\varepsilon}\right)^{-1},
\end{equation*}
we then get 
$\mu=\tilde c(Mb(\varepsilon))$ and we obtain the following lower bound for the optimal cost
\begin{equation*}
\Opt(\Gamma_M)\geq \min_{x\in\mathcal{F}_M}\sum_{i=1}^n x_ic((\alpha_i-\varepsilon)x_i)=Mc(\tilde{c}^{-1}(\mu))=Mc(Mb(\varepsilon)).
\end{equation*}
Combining the previous bounds we obtain the following estimate for the price of anarchy
\begin{equation*}
\PoA(\Gamma_M)\leq\frac{Mc(Ma(\varepsilon))}{Mc(Mb(\varepsilon))}.
\end{equation*}
Letting $M\to\infty$ and using the fact that $c$ is $\beta$-regularly varying we deduce
\begin{equation*}
\limsup_{M\to\infty}\PoA(\Gamma_M)\leq \left(\frac{a(\varepsilon)}{b(\varepsilon)}\right)^\beta
\end{equation*}
and since $a(\varepsilon)/b(\varepsilon)\to 1$ as $\varepsilon\to 0$ we conclude
\begin{equation*}
\limsup_{M\to\infty}\PoA(\Gamma_M)=1. 
\end{equation*} 

If some $\alpha_{i}=\infty$, then call $I_{0}:=\{i:\alpha_{i}<\infty\}$. In equilibrium
\begin{equation*}
M = \sum_{i=1}^{n}c_{i}^{-1}(\lambda) \ge \sum_{i\in I_{0}} c_{i}^{-1}(\lambda) \ge \sum_{i\in I_{0}}\frac{1}{\alpha_{i}+\varepsilon} c^{-1}(\lambda),
\end{equation*}
hence
\begin{equation*}
\lambda \le c\left(M \left(\sum_{i\in I_{0}} \frac{1}{\alpha_{i}+\varepsilon}\right)^{-1}\right).
\end{equation*}
In the optimum proceed as before with $\alpha_{i}' \nearrow \alpha_{i}$. \qed
\end{proof}
The following results follow easily from Theorem~\ref{th:regvar1}.

\begin{corollary}\label{co:somefinite}
In the game $\Gamma_{M}$ if for all $i \in E$ we have
$\lim_{x \to \infty} c_i(x)/x = m_i \in (0,+\infty]$ 
and at least one $m_{i}<\infty$, then
\begin{equation*}
\lim_{M \to \infty} \PoA(\Gamma_{M}) = 1.
\end{equation*}
\end{corollary}

\begin{corollary}\label{co:ciprime}
In the game $\Gamma_{M}$ if for all $i \in E$ we have
 $\lim_{x\to\infty}c_i^{\prime}(x) = m_i$ with $m_i \in (0,+\infty]$ and at least one $m_{i}$ is finite, then 
\begin{equation*}
\lim_{M \to \infty} \PoA(\Gamma_{M}) = 1.
\end{equation*}

\end{corollary}

\begin{corollary}\label{co:rvc-1}
In the game $\Gamma_{M}$ if for all $i \in E$ for some $\beta>0$ there exists a $\beta$-regularly varying function $c(\cdot)$ such that
\begin{equation}\label{eq:regvar_cor2}
\lim_{x \to \infty} \frac{c_i(x)}{c(x)} = m_i \in (0,+\infty],
\end{equation}
and at least one $m_{i}$ is finite, then 
\begin{equation*}
\lim_{M \to \infty} \PoA(\Gamma_{M}) = 1.
\end{equation*}
\end{corollary}

\begin{corollary}\label{co:affinecosts}
In the game $\Gamma_{M}$ if, for all $e_{i}\in E$, $c_{i}(x)=a_{i}+b_{i}x$, then
\begin{equation*}
\lim_{M \to \infty} \PoA(\Gamma_{M}) = 1.
\end{equation*}
\end{corollary}

\subsubsection{Costs bounded by affine functions}

The next theorem examines the case where each cost function is bounded above and below by two affine functions with the same slope, as in Figure~\ref{fi:affinebound}.

\begin{theorem}\label{th:affinebounds}
Consider the game $\Gamma_{M}$ and assume that for every $e_{i}\in E$
\begin{equation*}
\ell_{i}(x):=a_{i}+b_{i}x \le c_{i}(x) \le \alpha_{i}+b_{i}x=:L_{i}(x).
\end{equation*}
Then
\begin{equation*}
\lim_{M\to\infty}\PoA(\Gamma_{M}) = 1.
\end{equation*}
\end{theorem}

\section{Ill behaved games}\label{se:illbehaved}

In this section we will consider some examples where the price of anarchy is not asymptotic to one, as the inflow goes to infinity.

%
%
%

Consider a standard Pigou graph and assume that the costs are as follows: 
\begin{gather}\label{eq:cidstep}
\begin{aligned}
c_{1}(x)&=x, \\
c_{2}(x)&= a^{k} \quad\text{for }x \in (a^{k-1},a^{k}],\quad k\in\mathbb{Z},
\end{aligned}
\end{gather} 
with $a\ge 2$, as in Figure~\ref{fi:stepfunction}.
In this game the cost of one edge is the identity, whereas for the other edge it is a step function that touches the identity at intervals that grow exponentially.
The cost function $c_{2}$ is not continuous, but a very similar game can be constructed by approximating it with a continuous function.

\begin{figure}[h]
\centering
{%
\begin{tikzpicture}[scale=.6]

    \draw [thin, gray, ->] (0,-1) -- (0,10)      
        node [above, black] {$y$};              

    \draw [thin, gray, ->] (-1,0) -- (10,0)      
        node [right, black] {$x$};              

    \draw [draw=blue,thick,mylabel=at 0.7 below right with {$c_{e_1}(x)$}] (0,0) -- (10,10);

    \draw [draw=black,thick] (1/8,1/8) |- (1/4,1/4);
    \draw [draw=black,thick] (1/4,1/4) |- (1/2,1/2);
    \draw [draw=black,thick] (1/2,1/2) |- (1,1);    
    \draw [draw=black,thick] (1,1) |- (2,2);
    \draw [draw=black,thick] (2,2) |- (4,4);
    \draw [draw=black,thick,mylabel=at 0.7 above left with {$c_{e_2}(x)$}] (4,4) |- (8,8);

\begin{customlegend}[legend entries={$c_{e_2}(x)$,$c_{e_1}(x)$,C,$d$},legend style={at={(9,3.5)},font=\footnotesize}]
    \addlegendimage{black,sharp plot}
    \addlegendimage{blue,sharp plot}
    \end{customlegend}
\end{tikzpicture}
}
~\vspace{0cm} 
\caption{Step function.}
\label{fi:stepfunction} 
\end{figure}
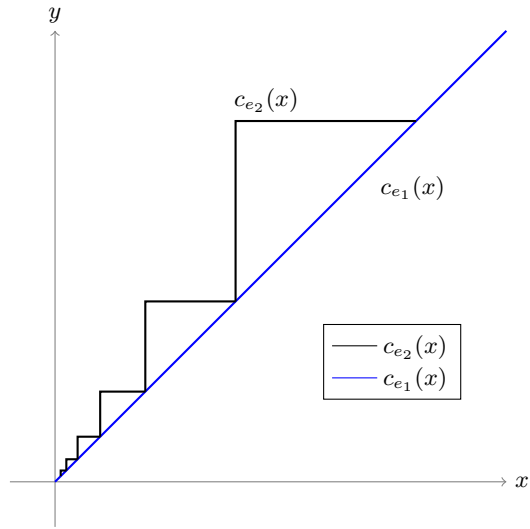

\begin{theorem}\label{th:PoAlimit}
Consider the game $\Gamma_{M}$ with costs as in \eqref{eq:cidstep}. We have
\begin{equation*}
\liminf_{M\to\infty}\PoA(\Gamma_{M})=1, \quad \limsup_{M\to\infty}\PoA(\Gamma_{M})=\frac{4+4a}{4+3a}.
\end{equation*}
\end{theorem}

\begin{remark}\label{re:limitto0}

We can immediately see that 
\begin{equation*}
\limsup_{M\to\infty}\PoA(\Gamma_{M})=\frac{6}{5}\quad\text{for }a=2
\end{equation*}
and 
\begin{equation*}
\limsup_{M\to\infty}\PoA(\Gamma_{M})\to\frac{4}{3}\quad\text{as }a\to\infty.
\end{equation*}

The proof of Theorem~\ref{th:PoAlimit} shows that there is a periodic behavior of the price of anarchy (on a logarithmic scale). This implies that 
\begin{equation*}
\liminf_{M\to 0}\PoA(\Gamma_{M})=1, \quad \limsup_{M\to 0}\PoA(\Gamma_{M})=\frac{4+4a}{4+3a}.
\end{equation*}
That is, even for very small values of $M$ the price of anarchy is not necessarily close to $1$.
\end{remark}

Figure~\ref{fi:PoAa=3} plots the price of anarchy for $M\in[2a^{k},2a^{k+1}]$, when $a=3$.

\begin{figure}[h]
\centering
    \includegraphics[width=0.5\textwidth]{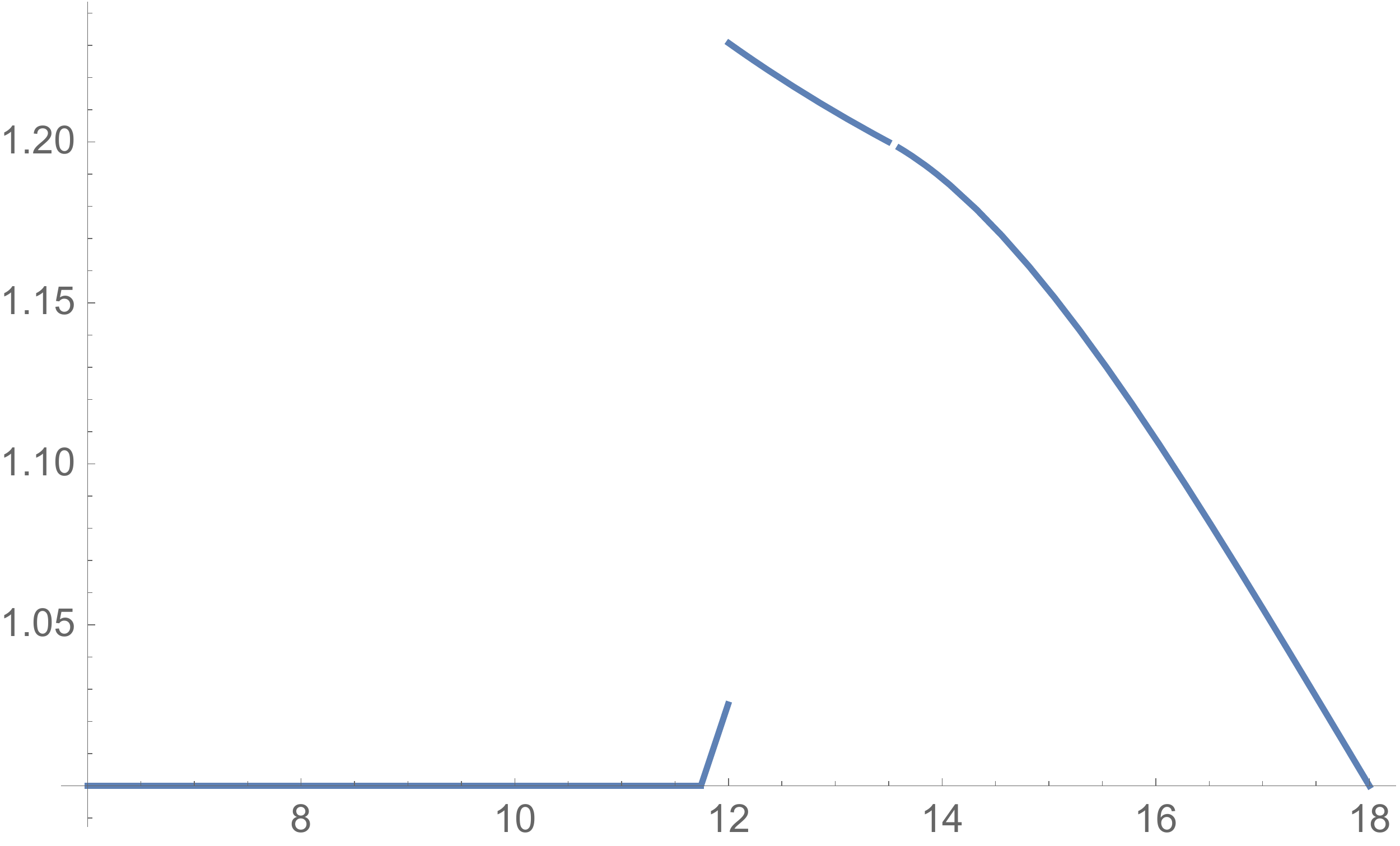}
    \caption{Price of anarchy for $M\in[2a^{k},2a^{k+1}]$, with $a=3$, $k=1$.}
\label{fi:PoAa=3}    
\end{figure}

The next theorem shows that the price of anarchy may fail to be asymptotic to one, even when the cost functions are all convex.

\begin{theorem}\label{th:illconvex}
There exist congestion games $\Gamma_{M}$ where the cost functions are all increasing and convex and both 
\begin{equation*}
\limsup_{M\to\infty}\PoA(\Gamma_{M})>1\quad\text{and}\quad\limsup_{M\to 0}\PoA(\Gamma_{M})>1.
\end{equation*}
\end{theorem}

The next theorem shows that the $\limsup$ of the price of anarchy may even be infinite.

\begin{theorem}\label{th:PoAtoinfty}
There exist congestion games $\Gamma_{M}$ where $\limsup_{M\to\infty}\PoA(\Gamma_{M})=\infty$.
\end{theorem}

%
%
%

\subsubsection*{Acknowledgments} 
Riccardo Colini-Baldeschi is a member of GNCS-INdAM.
Roberto Cominetti gratefully acknowledges the support and hospitality of LUISS during a visit in which this research was initiated. His research is also supported by N\'ucleo Milenio Informaci\'on y Coordinaci\'on en Redes ICM/FIC P10-024F.
Marco Scarsini is a member of GNAMPA-INdAM. His work is partially supported by PRIN  and MOE2013-T2-1-158.

\bibliographystyle{splncs03}
\bibliography{biblimitpoa}

\newpage

\appendix

\section{Regularly varying functions}\label{se:rvf}

The reader is referred to \cite{BinGolTeu:CUP1989} for an extended treatment of regularly varying functions. We study here some properties that are useful for our results.

\begin{lemma}\label{le:definitions_reg_var}

Let $\beta > 0$ and let $\Theta$ be a continuous and strictly increasing function, then the following definitions are equivalent:

\begin{enumerate}[{\rm (a)}]
\item\label{it:le:definitions_reg_var-a} 
the function $\Theta$ is $\beta$-regularly varying,
	
\item\label{it:le:definitions_reg_var-b} 
the function $\Theta^{-1}$ is $\frac{1}{\beta}$-regularly varying,
	
\item\label{it:le:definitions_reg_var-c}
for all $\gamma>0$ 
\begin{equation*}
\lim_{t \to \infty} \frac{1}{t}\Theta^{-1}(\gamma \Theta(t)) = \gamma^{1/\beta}.
\end{equation*}

\end{enumerate}

\end{lemma}

\begin{proof}
The equivalence of \ref{it:le:definitions_reg_var-a} and \ref{it:le:definitions_reg_var-b} is proved in 
\cite{deH:MCA1970} at page 22.

The equivalence of \ref{it:le:definitions_reg_var-b} and \ref{it:le:definitions_reg_var-c} is immediate, since, by setting $u = \Theta(t)$, we have
\begin{equation*}
\frac{1}{t} \Theta^{-1}(\gamma \cdot \Theta(t)) =
\frac{\Theta^{-1}(\gamma \cdot u)}{\Theta^{-1}(u)} \to \gamma^{1/\beta}. \qquad\qed
\end{equation*}
\end{proof}

\begin{lemma}\label{le:intANDprod_reg_var}
If $\Theta$ is a continuous and strictly increasing $\beta$-regularly varying function, then $x \cdot \Theta(x)$ and $\int_0^{x} \Theta(s) \diff s$ are $(1+\beta)$-regularly varying functions.
\end{lemma}

\begin{proof}
Let's first discuss the function $x \cdot \Theta(x)$. Observe that
\begin{equation*}
\lim_{x \to \infty} \frac{ax \Theta(ax)}{x \Theta(x)} = a \cdot a^{\beta} = a^{1+\beta}.
\end{equation*}
Similarly,
\begin{equation*}
\lim_{x \to \infty} \frac{\int_0^{ax} \Theta(s) \diff s}{\int_0^{x} \Theta(s) \diff s} = \lim_{x \to \infty} \frac{ \Theta(ax) a}{ \Theta(x) } = a^{\beta} a = a^{1+\beta}. \qquad\qed
\end{equation*}
\end{proof}

The following two lemmata appear in Proposition $1.5.7$ in \cite{BinGolTeu:CUP1989}.

\begin{lemma}\label{le:concate_reg_var}
For $i=1,2$, let $\Theta_{i}$ be a continuous and strictly increasing $\beta_{i}$-regularly varying function. Then $\Theta_1 \circ \Theta_2$ is $\beta_1 \cdot \beta_2$-regularly varying.
\end{lemma}

\begin{lemma}\label{lem:sum_reg_var}
Let $\Theta_1$ and $\Theta_2$ be two continuous and strictly increasing $\beta$-regularly varying functions, then
$\Theta_1 + \Theta_2$ is  $\beta$-regularly varying.
\end{lemma}

\section{Omitted proofs}\label{se:proofs}

\subsection*{Proofs of Section~\ref{se:well_behaved}}

\begin{proof}[of Theorem~\ref{th:addconstant}]
If some $c_i(\cdot)$ remains bounded the conclusion follows 
from Theorem~\ref{th:ConstantCost}, so we focus on the case where 
$c_i(x)\to\infty$ as $x\to\infty$ for all $i$.
In this case all the equilibrium flows $\eqflow{x}_i$ must diverge to $\infty$ as $M\to\infty$. In particular they will be all positive and
 the equilibrium is characterized by $c_i(\eqflow{x}_i)=\lambda$ for some $\lambda\to\infty$ as $M\to \infty$. 
In fact, since $\sum_{i=1}^n\eqflow{x}_i=M$ we can get $\lambda$ by solving the equation
$g(\lambda)=M$ where $g(\lambda)=\sum_{e_{i}\in E} c_{i}^{-1}(\lambda)$.

The same applies to  $\Gamma_{M}^{\boldsymbol{a}}$.
Call $\lambda^{\boldsymbol{a}}$ the cost at the equilibrium on each edge in $\Gamma_{M}^{\boldsymbol{a}}$
and $\boldsymbol{x}^{\boldsymbol{a}}$ the equilibrium of $\Gamma_{M}^{\boldsymbol{a}}$. Then we have
$a_{i}+c_{i}(x_{i}^{\boldsymbol{a}})=\lambda^{\boldsymbol{a}}$ so that
\begin{equation*}
M= \sum_{e_{i}\in E} c_{i}^{-1}(\lambda^{\boldsymbol{a}}-a_{i}).
\end{equation*}
Denoting $\underline{a}:=\min_{e_{i}\in E}a_{i}$ and $\bar{a}:\max_{e_{i}\in E}a_{i}$, the monotonicity of $c_i(\cdot)$
gives
\begin{equation*}
g(\lambda^{\boldsymbol{a}}-\bar{a})\leq M\leq g(\lambda^{\boldsymbol{a}}-\underline{a})
\end{equation*}
and since $M=g(\lambda)$ we get the inequality $\lambda^{\boldsymbol{a}}-\bar{a}\leq \lambda\leq \lambda^{\boldsymbol{a}}-\underline{a}$
which implies 
\begin{equation}\label{eq:lambda/lambda}
\lim_{M \to \infty} \frac{\lambda^{\boldsymbol{a}}}{\lambda} = 1.
\end{equation}

Now, for the optimum we have
\begin{equation*}
\Opt(\Gamma_{M}^{\boldsymbol{a}})=\min_{x\in\mathcal{F}_M}\sum_{e_{i}\in E} x_{i}(a_{i}+c_{i}(x_{i})) \ge \underline{a}M+\Opt(\Gamma_{M})
\end{equation*}
and we derive the estimate
\begin{equation*}
\PoA(\Gamma_{M}^{\boldsymbol{a}})
=\frac{M\lambda^{\boldsymbol{a}}}{\Opt(\Gamma_{M}^{\boldsymbol{a}})} 
\leq \frac{M\lambda^{\boldsymbol{a}}}{\underline{a}M+\Opt(\Gamma_M)}
= \frac{\lambda^{\boldsymbol{a}}/\lambda}{\underline{a}/\lambda+\frac{\Opt(\Gamma_{M})}{M\lambda}}\to 1,
\end{equation*}
which follows from the assumption $\Opt(\Gamma_{M})/(M\lambda)=\PoA(\Gamma_M)^{-1}\to 1$,
combined with \eqref{eq:lambda/lambda} and the fact that $\lambda\to\infty$.\qed
\end{proof}

\begin{proof}[of Corollary~\ref{co:somefinite}]
Apply Theorem~\ref{th:regvar1} with $c$ equal to the identity.
\end{proof}

\begin{proof}[of Corollary~\ref{co:ciprime}]
Just notice that $\lim_{x\to\infty}c_i^{\prime}(x) = m_i$ implies $\lim_{x \to \infty} c_i(x)/x = m_i$.
\end{proof}

\begin{proof}[of Corollary~\ref{co:rvc-1}]
If $m_{i}<\infty$, then from  \eqref{eq:regvar_cor2} we can derive the following inequalities:
\begin{equation*}
(m_{i}-\varepsilon)c(x) \le c_{i}(x) \le (m_{i}+\varepsilon)c(x),
\end{equation*}
hence
\begin{equation*}
\underset{\begin{matrix}
\downarrow\\
(m_{i}-\varepsilon)^{1/\beta}
\end{matrix}}
{\frac{c^{-1}((m_{i}-\varepsilon)c(x))}{x}} \le \frac{c^{-1}(c_{i}(x))}{x} \le \underset{\begin{matrix}
\downarrow\\
(m_{i}+\varepsilon)^{1/\beta}
\end{matrix}}{\frac{c^{-1}((m_{i}+\varepsilon)c(x))}{x}},
\end{equation*}
by Lemma~\ref{le:definitions_reg_var}\ref{it:le:definitions_reg_var-b}.
Since 
\begin{equation*}
(m_{i}-\varepsilon)^{1/\beta} \xrightarrow[\varepsilon\to 0]{} m_{i}^{1/\beta}\quad\text{and}\quad (m_{i}+\varepsilon)^{1/\beta} \xrightarrow[\varepsilon\to 0]{} m_{i}^{1/\beta},
\end{equation*}
if we call $\alpha_{i}=m_{i}^{1/\beta}$, we have
\begin{equation*}
\frac{c^{-1}(c_{i}(x))}{x} \to \alpha_{i}
\end{equation*}
and we can apply Theorem~\ref{th:regvar1}.

If $m_{i}=\infty$, then take 
\begin{equation*}
m'_{i} \le \frac{c_{i}(x)}{c(x)} \Longrightarrow m'_{i} c_{i}(x) \le c_{i}(x) \Longrightarrow \frac{c^{-1}(m'_{i} c_{i}(x))}{x} \le \frac{c^{-1}(c_{i}(x))}{x}.
\end{equation*}
Now
\begin{equation*}
\frac{c^{-1}(m'_{i} c_{i}(x))}{x} \xrightarrow[x\to\infty]{} (m'_{i})^{1/\beta}\xrightarrow[m'_{i}\to\infty]{}\infty = \alpha_{i} = \lim\frac{c^{-1}(c_{i}(x))}{x}.
\end{equation*}
and the previous result can be applied. \qed
\end{proof}

\begin{proof}[of Corollary~\ref{co:affinecosts}]
The conditions in Corollary~\ref{co:somefinite} are satified. \qed
\end{proof}

\begin{proof}[of Theorem~\ref{th:affinebounds}]
The results follows easily from Corollary~\ref{co:somefinite}. \qed
\end{proof}

\subsection*{Proofs of Section~\ref{se:illbehaved}}

In the whole subsection, for the sake of simplicity, we call $x$ the flow on $e_{1}$ and $y$ the flow on $e_{2}$.

\begin{proof}[of Theorem~\ref{th:PoAlimit}]
Let us study the price of anarchy for $M\in(2a^k, 2a^{k+1}]$.

\medskip
\noindent
\textbf{Equilibrium cost.}
In the subinterval $M\in(2a^{k}, a^{k}+a^{k+1}]$ we have 
\begin{align*}
\eqflow{x}&=M-a^{k}, & c_1(\eqflow{x})&= M-a^{k}\le a^{k+1},\\
\eqflow{y}&=a^{k}, & c_2(\eqflow{y})&= a^{k}.
\end{align*} 
For $M\in(a^{k}+a^{k+1},2a^{k+1}]$ we have 
\begin{align*}
\eqflow{x}&=a^{k+1}, & c_1(\eqflow{x})&= a^{k+1},\\
\eqflow{y}&=M-a^{k+1}, & c_2(\eqflow{y})&= a^{k+1}.
\end{align*} 
Therefore
\begin{equation*}
\WEq(\Gamma_M)=
\begin{cases}
(M-a^k)^2+a^{2k}&\text{for }M\in(2a^{k}, a^{k}+a^{k+1}],\\
Ma^{k+1}&\text{for }M\in(a^{k}+a^{k+1},2a^{k+1}].
\end{cases}
\end{equation*}

\medskip
\noindent
\textbf{Optimal cost.} In order to compute the optimal cost 
\begin{equation*}
\Opt(\Gamma_M)=\min_{0\leq y\leq M} yc_2(y)+(M-y)^2
\end{equation*}
we decompose the problem over the intervals $I_j=(a^j,a^{j+1}]$ on 
which $c_2(\cdot)$ is constant, namely, we consider the 
subproblems
\begin{equation*}
C_j=\min_{y\in I_j,y\leq M}a^{j+1}y+(M-y)^2.
\end{equation*}
We observe that for $j\geq k+2$ we have $a^j\geq a^{k+2}\geq 2a^{k+1}\geq M$ so that  
$C_j$ is infeasible and therefore 
$\Opt(\Gamma_M)=\min\{C_0,C_1,\ldots,C_{k+1}\}$.
In fact, we will show that $\Opt(\Gamma_M)=\min\{C_{k-1},C_k\}$.

Let us compute $C_j$. Since $(M-y)^2$
is symmetric around $M$, the constraint $y\leq M$ can be dropped
and then the minimum $C_j$ is obtained by projecting 
onto $[a^{j},a^{j+1}]$ the unconstrained minimizer $y_{j}=M-a^{j+1}/2$.
We get
\begin{equation}\label{CJ}
C_j=
\begin{cases}
a^{j+1}a^j+(M-a^j)^2&\text{ if }M<a^j+\frac{a^{j+1}}{2},\\
a^{j+1}(M-\frac{a^{j+1}}{2})+(\frac{a^{j+1}}{2})^2&\text{ if }a^j+\frac{a^{j+1}}{2}\leq M\leq a^{j+1}+\frac{a^{j+1}}{2},\\
a^{j+1}a^{j+1}+(M-a^{j+1})^2&\text{ if }M>a^{j+1}+\frac{a^{j+1}}{2}.
\end{cases}
\end{equation}

\begin{newclaim}\label{cl:cj-1>cj}
For $j\le k-1$ we have $C_{j} = a^{j+1}a^{j+1}+(M-a^{j+1})^{2}$ and 
$C_{j-1} \ge C_{j}$. 
\end{newclaim}

\begin{proof}
The expression for $C_j$ follows from \eqref{CJ} if we note that $M > 2a^{k} \ge \frac{3}{2}a^{j+1}$.
In order to prove that $C_{j-1} \ge C_{j}$ we observe that
\begin{align*}
C_{j-1} \ge C_{j} &\iff (a^{j})^{2}+(M-a^{j})^{2} \ge (a^{j})^{2}a^{2}+(M-a^{j}a)^{2}\\
&\iff 2(a^{j})^{2} + M^{2} - 2Ma^{j} \ge 2(a^{j})^{2} a^{2} + M^{2} -  2Ma^{j} a \\
&\iff Ma^{j}(a-1) \ge (a^{j})^{2}(a^{2}-1) \\
&\iff M \ge a^{j}(a+1) = a^{j}+a^{j+1}.
\end{align*}
Since $M>2a^{k}=a^{k}+a^{k}\ge a^{j}+a^{j+1}$, this holds true. \qed
\end{proof}

\begin{newclaim}\label{cl:cj<cj+1}
$C_{k+1}=a^{k+2}a^{k+1}+(M-a^{k+1})^{2}\geq C_{k-1}.$
\end{newclaim}

\begin{proof}

Since $M \le 2 a^{k+1} \le a^{k+1}+\frac{a^{k+2}}{2}$ we get 
the expression for $C_{k+1}$ from \eqref{CJ}.
Then
\begin{align*}
C_{k-1} \le C_{k+1} &\iff (a^{k})^{2} + (M-a^{k})^{2} \le a^{k+2}a^{k+1} + (M-a^{k+1})^{2} \\
&\iff 2(a^{k})^{2} + M^{2} - 2Ma^{k} \le (a^{k})^{2}a^{3} + (a^{k})^{2}a^{2} + M^{2} - 2Ma^{k+1} \\
&\iff 2Ma^{k}(a-1) \le (a^{k})^{2}(a-1)(a^{2}+2a
+2) \\
&\iff 2M \le a^{k}(a^{2}+2a+2).
\end{align*}
Since $M \le 2a^{k+1}$ it suffices to have $4a^{k+1}\le a^{k}(a^{2}+2a+2)$ which is easily seen to hold.
\qed
\end{proof}

Combining the previous claims we get that $\Opt(\Gamma_M)=\min\{C_{k-1},C_k\}$.
It remains to figure out which one between $C_{k-1}$ and $C_k$ attains the minimum.
This depends on where $M$ is located within the interval $(2a^k,2a^{k+1}]$ as
explained in our next claim. In the sequel we denote
\begin{align*}
\alpha&=1+\frac{a}{2}\\ 
\beta&=1+\frac{a}{2}+\sqrt{a-1}\\
\gamma&=\frac{3}{2}a
\end{align*}
and we observe that 
\begin{equation*}
2\leq\alpha\leq\beta\leq\gamma\leq 2a.
\end{equation*}

\begin{newclaim}\label{cl:ck<ck-1}
For $M\in (2a^k,2a^{k+1}]$ we have 
\begin{equation*}
\Opt(\Gamma_M)=
\begin{cases}
C_{k-1}= (a^{k})^{2} + (M-a^{k})^{2}&\text{ if } M\in (2a^k,\alpha a^k)\\
C_{k-1}= (a^{k})^{2} + (M-a^{k})^{2}&\text{ if } M\in [\alpha a^k,\beta a^k)\\
C_{k}= a^{k+1}(M-\frac{1}{4}a^{k+1})&\text{ if } M\in [\beta a^k,\gamma a^k]\\
C_{k}= (a^{k+1})^{2} + (M-a^{k+1})^{2}&\text{ if } M\in (\gamma a^k,2 a^{k+1}].
\end{cases}
\end{equation*}
\end{newclaim}
\begin{proof}
From \eqref{CJ} we have $C_{k-1}= (a^{k})^{2} + (M-a^{k})^{2}$ whereas the expression
for $C_k$ changes depending where $M$ is located.

\medskip
\noindent
\textbf{(a) Initial interval $M\in(2a^k,\alpha a^k)$}. 

\noindent
Here $M<a^k+\frac{1}{2}a^{k+1}$ so that
\eqref{CJ} gives $C_k=a^{k+1}a^{k} + (M-a^{k})^{2}$. Hence, clearly $C_{k-1}\le C_{k}$
and $\Opt(\Gamma_M)=C_{k-1}$.

\medskip
\noindent
\textbf{(b) Final interval $M\in (\gamma a^{k},2a^{k+1}]$}. 

\noindent
Here $M>\gamma a^k =\frac{3}{2}a^{k+1}$ 
so that \eqref{CJ} gives $C_{k}=(a^{k+1})^{2}+(M-a^{k+1})^{2}$. Proceeding
as in the proof of Claim \ref{cl:cj-1>cj}, we have $C_{k-1}\ge C_{k}$ if and only if $M \ge a^{k}+a^{k+1}$.
The latter holds since  $M \ge \frac{3}{2}a^{k+1} \ge a^{k+1}+a^{k}$. Hence
$\Opt(\Gamma_M)=C_{k}$.

\medskip
\noindent
\textbf{(c) Intermediate interval $M\in [\alpha a^k,\gamma a^{k}]$}. 

\noindent
Here $a^{k}+\frac{1}{2}a^{k+1}\leq M \leq \frac{3}{2}a^{k+1}$ so that  \eqref{CJ}
gives
\begin{equation*}
C_{k}= a^{k+1}\left(M-\frac{1}{2}a^{k+1}\right)+\left(\frac{1}{2}a^{k+1}\right)^{2}=a^{k+1}\left(M-\frac{1}{4}a^{k+1}\right).
\end{equation*}
Then, denoting $z=M/a^k$ we have
\begin{align*}
C_{k-1}\le C_{k} &\iff 2(a^{k})^{2} + M^{2} - 2Ma^{k} \le a^{k+1}M -\left(\frac{1}{2}a^{k+1}\right)^{2}.\\
&\iff z^2-z(2+a)+\left(2+\frac{1}{4}a^2\right)\leq 0\\
&\iff 1+\frac{1}{2}a-\sqrt{a-1} \leq z \leq 1+\frac{1}{2}a+\sqrt{a-1}.
\end{align*}
The upper limit for $z$ is precisely $\beta$ while the lower limit is smaller than $\alpha$.
Hence $\Opt(\Gamma_M)=C_{k-1}$ for $M\in[\alpha a^k,\beta a^k]$ 
and $\Opt(\Gamma_M)=C_{k}$ for $M\in[\beta a^k,\gamma a^k]$.
\qed
\end{proof}

\begin{figure}[h]
\centering
{%
\begin{tikzpicture}[scale=.7]
\filldraw 
(0,0) circle (2pt) node[align=left,   below] {$2a^k$} --
(4,0) circle (2pt) node[align=center, below] {$\beta a^k$}     -- 
(8,0) circle (2pt) node[align=center, below] {~~~$a^k\!+\!a^{k+1}$}     -- 
(12,0) circle (2pt) node[align=right,  below] {~~~$\frac{3}{2}a^{k+1}$} --
(16,0) circle (2pt) node[align=right,  below] {~~~$2a^{k+1}$};
\draw[decorate, decoration={brace, amplitude=10pt}] (3.9,-0.7) -- coordinate [left=10pt] (A) (0,-0.7) node {};
\draw[decorate, decoration={brace, amplitude=10pt}] (11.9,-0.7) -- coordinate [left=10pt] (B) (4.1,-0.7) node {};
\draw[decorate, decoration={brace, amplitude=10pt}] (16,-0.7) -- coordinate [left=10pt] (C) (12.1,-0.7) node {};
\node (D) at (2,-1.6) {\scriptsize $(a^k)^2+(M\!-\!a^k)^2$};
\node (E) at (8,-1.6) {\scriptsize $a^{k+1}(M\!-\!\frac{a^{k+1}}{4})$};
\node (O) at (8,-2.6) {\bf $\Opt(\Gamma_M)$};
\node (F) at (14,-1.6) {\scriptsize $(a^{k+1})^2+(M\!-\!a^{k+1})^2$};
\draw[decorate, decoration={brace, amplitude=10pt}] (0,0.3) -- coordinate [left=10pt] (G) (7.9,0.3) node {};
\draw[decorate, decoration={brace, amplitude=10pt}] (8.1,0.3) -- coordinate [left=10pt] (H) (16,0.3) node {};
\node (I) at (4,1.2) {\scriptsize $(a^k)^2+(M\!-\!a^k)^2$};
\node (J) at (12,1.2) {\scriptsize $a^{k+1}M$};
\node (W) at (8,2.1) {\bf $\WEq(\Gamma_M)$};
\end{tikzpicture}
}
~\vspace{0cm} 
\caption{Breakpoints for optimum and equilibrium.}
\label{fi:breakpoints}
\end{figure}
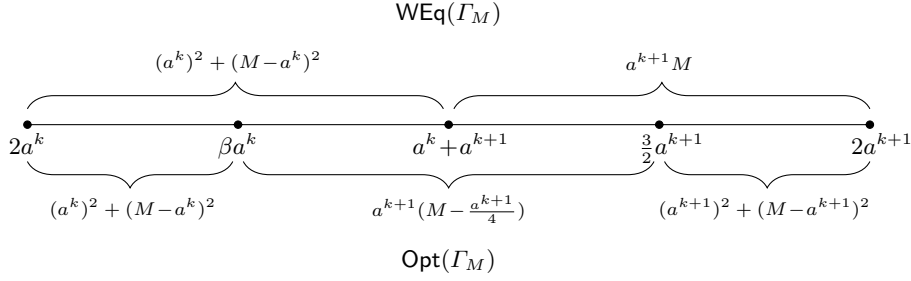

Figure~\ref{fi:breakpoints} illustrates the different intervals in which 
the equilibrium (above) and the optimum (below) change. Notice that $\Opt(\Gamma_M)$ 
varies continuously even at breakpoints, whereas $\WEq(\Gamma_M)$ has a jump at $a^k+a^{k+1}$.
We now proceed to examine the price of anarchy which will be expressed 
as a function of $z=M/a^k$. 

From the expressions of $\WEq(\Gamma_M)$ and $\Opt(\Gamma_M)$
(see Figure~\ref{fi:breakpoints}) it follows that $\PoA(\Gamma_M)=1$ throughout 
the initial interval $M\in (2a^k,\beta a^k)$. 
Over the next interval $M\in[\beta a^k,a^k\!+\!a^{k+1}]$ we have
\begin{equation*}
\PoA(\Gamma_M)=\frac{(a^k)^2+(M-a^k)^2}{a^{k+1}(M-a^{k+1}/4)}=\frac{1+(z-1)^2}{a(z-a/4)},
\end{equation*}
which increases from 1 at $z=\beta$ up to $(4+4a^2)/(a(4+3a))$ at $z=1+a$.

At $M=a^k+a^{k+1}$ the equilibrium has a discontinuity and $\PoA(\Gamma_M)$ 
jumps to $(4+4a)/(4+3a)$ and then it decreases over the interval 
$M\in (a^k+a^{k+1},\frac{3}{2}a^{k+1})$ as
\begin{equation*}
\PoA(\Gamma_M)=\frac{a^{k+1}M}{a^{k+1}(M-a^{k+1}/4)}=\frac{z}{z-a/4}.
\end{equation*}
Finally, for $M\in(\frac{3}{2}a^{k+1},2a^{k+1}]$ the price of anarchy continues 
to decrease as
\begin{equation*}
\PoA(\Gamma_M)=\frac{a^{k+1}M}{(a^{k+1})^2+(M-a^{k+1})^2}=\frac{az}{a^2+(z-a)^2}.
\end{equation*}
going back to 1 at $z=2a$ which corresponds to $M=2a^{k+1}$.

Thus the price of anarchy oscillates over each interval $(2a^k,2a^{k+1}]$ between 
a minimum value of 1 and a maximum of $(4+4a)/(4+3a)$. This completes the proof
of Theorem~\ref{th:PoAlimit}.
\qed
\end{proof}

\vspace{2ex}
\begin{proof}[of Theorem~\ref{th:illconvex}]
Consider a parallel network with two edges with a quadratic cost $c_1(x)=x^2$ on the upper edge and a lower edge cost defined by linearly interpolating $c_1$, that is, for $a\geq 2$ we let  (see Figure~\ref{fi:xsquare})
\begin{equation*}
c_2(y)=(a^{k-1}\!+\!a^k)y-a^{k-1}a^k,\quad \text{for }y\in[a^{k-1},a^k],\quad k\in\mathbb{Z}.
\end{equation*}

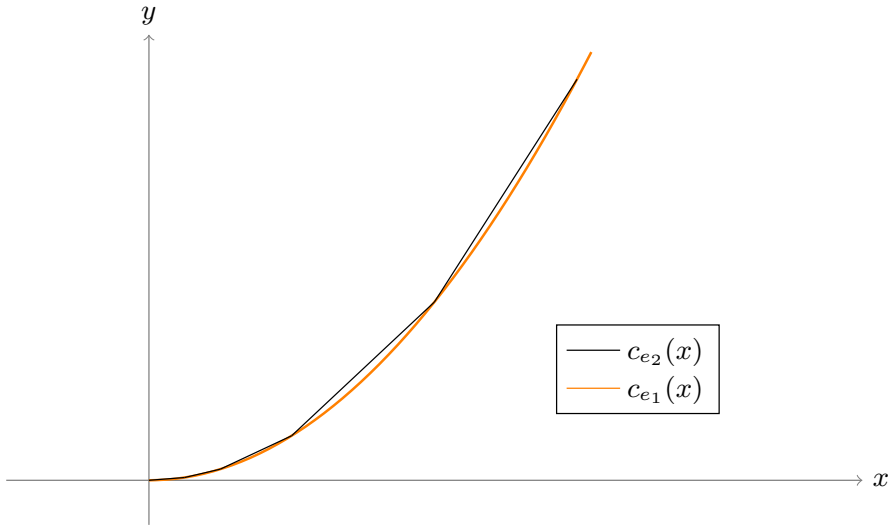
\begin{figure}[h]
\centering
\resizebox{!}{200pt}
{%
\begin{tikzpicture}[xscale=1.6,yscale=.5]
    \draw [thin, gray, ->] (0,-1) -- (0,10)      
        node [above, black] {$y$};              

    \draw [thin, gray, ->] (-1,0) -- (5,0)      
        node [right, black] {$x$};              

    \draw [color=orange,thick,domain=0:3.1,smooth,variable=\x] plot (\x,{(\x)^2});
    \draw [draw=black] (0.0,0.0) -- (0.25,0.0625);
    \draw [draw=black] (0.25,0.0625) -- (0.5,0.25);
    \draw [draw=black] (0.5,0.25) -- (1,1);
    \draw [draw=black] (1,1) -- (2,4);
    \draw [draw=black] (2,4) -- (3,9);
    \begin{customlegend}[legend entries={$c_{e_2}(x)$,$c_{e_1}(x)$,C,$d$},legend style={at={(4,3.5)},font=\footnotesize}]
    \addlegendimage{black,sharp plot}
    \addlegendimage{orange,sharp plot}
    \end{customlegend}
\end{tikzpicture}
}
~\vspace{0cm} 
\caption{$x^2$ and its linear interpolation}
\label{fi:xsquare} 
\end{figure}

Note that $c_1$ and $c_2$ are convex.
Consider the optimal cost problem

\begin{equation*}
\Opt(\Gamma_M)=\min_{\substack{x+y=M \\ x,y\geq 0}}x^3+y c_2(y).
\end{equation*}
Since the function $h(y)=y c_2(y)$ is non-differentiable, the optimality condition reads $3x^2\in\partial h(y)$. In particular, the subdifferential at $y=a^k$ is
\begin{equation*}
\partial h(a^k)=[a^{2(k-1)}(2a^2\!+\!a),a^{2k}(2a^2\!+\!a)]
\end{equation*}
and there is a range of values of $M$ for which the optimal solution is $y=a^{k}$. The smallest such $M$ is obtained when $3x^2=a^{2(k-1)}(2a^2\!+\!a)$. This gives as optimal solution $y=a^k$ and $x=a^{k-1} b$, with $b=\sqrt{(2a^2\!+\!a)/3}$, corresponding to $M_k=a^{k-1}[a+b]$ with optimal value
\begin{equation*}
\Opt(\Gamma_{M_k})=a^{3(k-1)}[b^3+a^3].
\end{equation*}

In order to find the equilibrium for $M_k$ we solve the equation
$x^2=c_2(y)$ with $x+y=M_k$. A routine calculation gives
$x=a^{k-1}c$ and $y=a^{k-1}d$ with
\begin{align*}
c=&\frac{1}{2}\left[\sqrt{(a+1)^2+4a^2+4(a+1)b}-(a+1)\right],\\
d=&a+b-c.
\end{align*}
Note that $1<d<a$ so that $y\in(a^{k-1},a^k)$, and therefore the 
equilibrium cost is
\begin{equation*}
\WEq(\Gamma_{M_k})=a^{3(k-1)}[c^3+(a+1)d^2-ad].
\end{equation*}

Putting together the previous formulas we get
\begin{equation*}
\PoA(\Gamma_{M_k})=\frac{c^3+(a+1)d^2-ad}{b^3+a^3}.
\end{equation*}
For $a=2$ this expression evaluates to $\PoA(\Gamma_{M_k})\sim 1.0059$
from which the result follows.
\qed
\end{proof}

\begin{proof}[of Theorem~\ref{th:PoAtoinfty}]
Consider a game $\Gamma_{M}=(\mathcal{G},M,\boldsymbol{c})$, where $\mathcal{G}=(V,E)$ with  $V=\{s,d\}$ and $E=\{e_{1},e_{2}\}$.
Take a sequence $\{\alpha_{k}\}$ such that $\alpha_{k+1}/\alpha_{k}\to\infty$ and assume that the costs are
\begin{align*}
c_{1}(x)&=c(x):=
\begin{cases}
\expon{}&\text{ for } x<1,\\
\expon{x}/x&\text{ for } x\ge 1,
\end{cases}\\
c_{2}(y)&=\bar{c}(y):=c(\alpha_{k+1})\quad\text{for }y\in (\alpha_{k},\alpha_{k+1}].
\end{align*}
Since we are interested in asymptotic results, we are concerned only with the case $c(x)=\expon{x}/x$.

In equilibrium
\begin{align*}
2\alpha_{k} < M \le \alpha_{k}+\alpha_{k+1} &\Longrightarrow \eqflow{y}=\alpha_{k}, \eqflow{x}=M-\alpha_{k},\\
\alpha_{k}+\alpha_{k+1} < M \le 2\alpha_{k+1} &\Longrightarrow \eqflow{y}=M-\alpha_{k+1}, \eqflow{x}=\alpha_{k+1}.
\end{align*}
At $M=\alpha_{k}+\alpha_{k+1}+\varepsilon$ we have 
\begin{equation*}
\WEq(\Gamma_{M})=c(\alpha_{k+1})\alpha_{k+1}+(M-\alpha_{k+1})c(\alpha_{k+1})=M c(\alpha_{k+1}).
\end{equation*}

We now turn to computing the optimum.
\begin{equation*}
\optflow{c}=\min_{0 \le x \le M}xc(x)+(M-x)\bar{c}(M-x)=\min_{j}\min_{\alpha_{j}<M-x\le\alpha_{j+1}}xc(x)+(M-x)c(\alpha_{j+1}).
\end{equation*}
The unconstrained optimization yields
\begin{align*}
\expon{x}&=\frac{\expon{\alpha_{j+1}}}{\alpha_{j+1}},\\
x_{j}&=\alpha_{j+1}-\ln\alpha_{j+1},\\
y_{j}&=M-\alpha_{j+1}+\ln\alpha_{j+1}.
\end{align*}
For the constrained minimizer we have
\begin{align*}
y_{j} \le \alpha_{j} &\Longrightarrow \optflow{y}_{j} = \alpha_{j} \Longrightarrow c_{j}=\expon{M-\alpha_{j}} + \frac{\alpha_{j}}{\alpha_{j+1}}\expon{\alpha_{j+1}},\\
y_{j} > \alpha_{j+1} &\Longrightarrow \optflow{y}_{j} = \alpha_{j+1} \Longrightarrow c_{j}=\expon{M-\alpha_{j+1}} + \expon{\alpha_{j+1}},\\
\alpha_{j} < y_{j} \le \alpha_{j+1} &\Longrightarrow \optflow{y}_{j} = y_{j} \Longrightarrow c_{j}=\expon{\alpha_{j+1}-\ln \alpha_{j+1}} + (M-\alpha_{j+1} + \ln \alpha_{j+1})\frac{\expon{\alpha_{j+1}}}{\alpha_{j+1}}\\
&\phantom{\Longrightarrow \optflow{y}_{j} = y_{j} \Longrightarrow c_{j}\ }=\frac{\expon{\alpha_{j+1}}}{\alpha_{j+1}}(1+M-\alpha_{j+1}+\ln \alpha_{j+1}).
\end{align*}
We consider two cases:

\noindent\textbf{Case 1}.
\begin{align*}
j\le k-1 \Longrightarrow M>2 \alpha_{j+1} &\Longrightarrow y_{j}=M - \alpha_{j+1} + \ln \alpha_{j+1} > \alpha_{j+1} + \ln \alpha_{j+1} > \alpha_{j+1} \\&\Longrightarrow c_{j}=\expon{M-\alpha_{j+1}} + \expon{\alpha_{j+1}}.
\end{align*}
\begin{align}
c_{j} \le c_{j-1} \iff h(\alpha_{j+1}) \le h(\alpha_{j}) \quad\text{where}\quad &h(x):= \expon{M-x}+\expon{x} \label{eq:h}\\
&h'(x)= \expon{x}-\expon{M-x}\le 0 \iff x \le \frac{M}{2} \nonumber\\
&\alpha_{j} \le \alpha_{j+1} \le \alpha_{k} \le \frac{M}{2} \Longrightarrow \text{ ok!}\nonumber
\end{align}

\noindent\textbf{Case 2}. 
Does $j\ge k+1$ imply $y_{j} \le \alpha_{j}$? Notice that
\begin{align*}
y_{j} \le \alpha_{j} &\iff M - \alpha_{j+1} + \ln \alpha_{j+1} \le \alpha_{j}\\
&\iff M \le \alpha_{j} + \alpha_{j+1} - \ln \alpha_{j+1}.
\end{align*}
Notice that
\begin{equation*}
M \le 2 \alpha_{k+1} \le 2 \alpha_{j}.
\end{equation*}
In turn the following mutual implications hold:
\begin{align}
2 \alpha_{j} &\le \alpha_{j} + \alpha_{j+1} - \ln \alpha_{j+1} \label{eq:2aj<ajaj+1ln}\\
&\Updownarrow \nonumber\\
\ln \alpha_{j+1} &\le \alpha_{j+1} - \alpha_{j} \nonumber\\
&\Updownarrow \nonumber\\
\frac{\ln \alpha_{j+1}}{\alpha_{j+1}} &\le 1 - \frac{\alpha_{j}}{\alpha_{j+1}} \label{eq:lnaj+1}.
\end{align}
Since 
\begin{equation*}
M\to\infty \Longrightarrow k\to\infty \Longrightarrow j\to\infty \Longrightarrow \frac{\ln \alpha_{j+1}}{\alpha_{j+1}} \to 0 \quad\text{and}\quad \frac{\alpha_{j}}{\alpha_{j+1}} \to 0.
\end{equation*}
we have that \eqref{eq:lnaj+1} and hence \eqref{eq:2aj<ajaj+1ln} hold asymptotically, therefore, for $M$ large,
\begin{equation*}
j \ge k+1 \Longrightarrow c_{j}=\expon{M- \alpha_{j}}+\frac{\alpha_{j}}{\alpha_{j+1}}\expon{\alpha_{j+1}}.
\end{equation*}

Does $j \ge k+1$ imply $c_{j}\le c_{j+1}$? Notice that
\begin{align*}
c_{j} \le c_{j+1} \iff \expon{M- \alpha_{j}}+\frac{\alpha_{j}}{\alpha_{j+1}}\expon{\alpha_{j+1}} \le \expon{M- \alpha_{j+1}}+\frac{\alpha_{j+1}}{\alpha_{j+2}}\expon{\alpha_{j+2}}.
\end{align*}
For $x \ge 1$ the function $x\mapsto h(x)$ defined as in \eqref{eq:h} is increasing ($h'(x)=x^{-2}(\expon{x}x-\expon{x}) \ge 0$ for $x \ge 1$). 
Therefore
\begin{equation*}
\expon{M- \alpha_{j}}+\alpha_{j}\frac{\expon{\alpha_{j+1}}}{\alpha_{j+1}} \le \expon{M- \alpha_{j+1}}+\alpha_{j+1}\frac{\expon{\alpha_{j+2}}}{\alpha_{j+2}}.
\end{equation*}
We now prove that the function 
\begin{equation*}
g(x):=\expon{M-x}+x\frac{\expon{\alpha_{j+2}}}{\alpha_{j+2}}
\end{equation*}
is increasing for $x \ge \alpha_{k+1}$. Indeed
\begin{equation*}
g'(x)=\frac{\expon{\alpha_{j+2}}}{\alpha_{j+2}}-\expon{M-x}
\end{equation*}
and
\begin{align*}
g'(x) \ge 0 &\iff M-x \le \alpha_{j+2} - \ln \alpha_{j+2} \\
&\iff M \le x + \alpha_{j+2} - \ln \alpha_{j+2},
\end{align*}
which is true for $x \ge \alpha_{k+1}$ iff $M \le \alpha_{k+1} + \alpha_{j+2} - \ln \alpha_{j+2}$.
Since $M \le 2 \alpha_{k+1}$, it is enough to prove that $\alpha_{k+1} \le  \alpha_{j+2} - \ln \alpha_{j+2}$, and since $x \mapsto x - \ln x$ is increasing for $x\ge 1$, it suffices to prove that $\alpha_{k+1} \le  \alpha_{k+2} - \ln \alpha_{k+2}$,  that is,
\begin{equation}\label{eq:ak+1k+2}
\frac{\alpha_{k+1}}{\alpha_{k+2}} \le 1 - \frac{\ln \alpha_{k+2}}{\alpha_{k+2}}.
\end{equation}
Since
\begin{equation*}
\frac{\alpha_{k+1}}{\alpha_{k+2}}\to 0\quad\text{and}\quad \frac{\ln \alpha_{k+2}}{\alpha_{k+2}} \to 0, \quad\text{as }k\to\infty,
\end{equation*}
the inequality \eqref{eq:ak+1k+2} holds for $k$ large enough.
More explicitly, since
\begin{equation*}
\frac{\ln \alpha_{k+2}}{\alpha_{k+2}} \le \max\frac{\ln x}{x}=\frac{1}{\expon{}},
\end{equation*}
it is enough to have
\begin{equation*}
\frac{\alpha_{k+1}}{\alpha_{k+2}} \le 1 - \frac{1}{\expon{}}.
\end{equation*}

As a consequence, 
\begin{align*}
\optflow{c}&=\min\{c_{k-1},c_{k},c_{k+1}\},\\
c_{k-1}&=\expon{M-\alpha_{k}}+\expon{\alpha_{k}},\\
c_{k+1}&=\expon{M-\alpha_{k+1}} + \alpha_{k+1}\frac{\expon{\alpha_{k+2}}}{\alpha_{k+2}},\\
c_{k}&\text{ depends on where }y_{k}\text{ lies w.r.t. the interval }(\alpha_{k},\alpha_{k+1}].
\end{align*}
If we call
\begin{equation*}
H(x):=\expon{M-x}+x \frac{\expon{\alpha_{k+2}}}{\alpha_{k+2}},
\end{equation*}
we have
\begin{align*}
y_{k} > \alpha_{k+1} &\Longrightarrow c_{k} = 
\expon{M-\alpha_{k+1}}+\expon{\alpha_{k+1}} = 
\expon{M-\alpha_{k+1}}+ \alpha_{k+1} \frac{\expon{\alpha_{k+1}}}{\alpha_{k+1}} \le H(\alpha_{k+1}) 
=c_{k+1}, \\
y_{k} \le \alpha_{k} &\Longrightarrow c_{k} = 
\expon{M-\alpha_{k}}+\frac{\alpha_{k}}{\alpha_{k+1}}\expon{\alpha_{k+1}}.
\end{align*}
\begin{newclaim}
If $y_{k} \le \alpha_{k}$, then $c_{k} \le c_{k+1}$ for $k$ large enough.
\end{newclaim}

\begin{proof}
\begin{equation*}
\expon{M-\alpha_{k}}+\alpha_{k}\frac{\expon{\alpha_{k+1}}}{\alpha_{k+1}} 
\le \expon{M-\alpha_{k}}+\alpha_{k}\frac{\expon{\alpha_{k+2}}}{\alpha_{k+2}} = H(\alpha_{k}).
\end{equation*}

\begin{align*}
H(x)\ge 0 \iff \frac{\expon{\alpha_{k+2}}}{\alpha_{k+2}} \ge \expon{M-x}
\iff \alpha_{k+2} - \ln \alpha_{k+2} +x \ge M.
\end{align*}
For $x= \alpha_{k}$ we have, for $k$ large enough,
\begin{equation*}
\alpha_{k+2} - \ln \alpha_{k+2} + a_{k} \ge 2 \alpha_{k+1} \ge M,
\end{equation*} 
since, dividing by $\alpha_{k+2}$,
\begin{equation*}
1 - \underset{\displaystyle{\begin{matrix}
\downarrow\\
0
\end{matrix}
}}
{\frac{\ln \alpha_{k+2}}{\alpha_{k+2}}} + 
\underset{\displaystyle{\begin{matrix}
\downarrow\\
0
\end{matrix}
}}
{\frac{a_{k}}{a_{k+2}}} \ge 
\underset{\displaystyle{\begin{matrix}
\downarrow\\
0
\end{matrix}
}}
{2\frac{a_{k+1}}{a_{k+2}}}.
\end{equation*}

\end{proof}

\begin{equation*}
\alpha_{k} < y_{k} \le \alpha_{k+1} \Longrightarrow c_{k} = 
\frac{\expon{\alpha_{k+1}}}{\alpha_{k+1}}(1+M- \alpha_{k+1} + \ln \alpha_{k+1}).
\end{equation*}
We need to prove that
\begin{equation*}
c_{k} \le 
c_{k+1}=\expon{M- \alpha_{k+1}} + \alpha_{k+1} \frac{\expon{\alpha_{k+2}}}{\alpha_{k+2}}.
\end{equation*}
Since $M \le 2\alpha_{k+1}$, we have 
\begin{equation*}
c_{k} \le \frac{\expon{\alpha_{k+1}}}{\alpha_{k+1}}(1+\alpha_{k+1} + \ln \alpha_{k+1}).
\end{equation*}
Now
\begin{equation*}
y_{k} > \alpha_{k} \Longrightarrow M- \alpha_{k+1} + \ln \alpha_{k+1} > \alpha_{k} \Longrightarrow \expon{M-\alpha_{k+1}} > \frac{\expon{\alpha_{k}}}{\alpha_{k+1}} \Longrightarrow c_{k+1} \ge \frac{\expon{\alpha_{k}}}{\alpha_{k+1}} + \alpha_{k+1} \frac{\expon{\alpha_{k+2}}}{\alpha_{k+2}}.
\end{equation*}
If we prove that
\begin{equation*}
\frac{\expon{\alpha_{k+1}}}{\alpha_{k+1}}(1+M- \alpha_{k+1} + \ln \alpha_{k+1}) \le 
\expon{M- \alpha_{k+1}} + \alpha_{k+1} \frac{\expon{\alpha_{k+2}}}{\alpha_{k+2}},
\end{equation*}
then we have the result. Notice that
\begin{equation*}
\underset{\begin{matrix}
\downarrow\\
1
\end{matrix}}
{\underbrace{\frac{1}{\alpha_{k+1}}(1 + \alpha_{k+1} + \ln \alpha_{k+1} - \expon{\alpha_{k}-\alpha_{k+1}})}} \le \underset{\displaystyle{\theta_{k}}}{\underbrace{\frac{\alpha_{k+1}}{\alpha_{k+2}}\expon{\alpha_{k+2}-\alpha_{k+1}}}}.
\end{equation*}
Now
\begin{align*}
\ln \theta_{k} &= (\alpha_{k+2} - \ln \alpha_{k+2}) - (\alpha_{k+1} - \ln \alpha_{k+1}) \\
&= \alpha_{k+2} \left(1 - \frac{\ln \alpha_{k+2}}{\alpha_{k+2}} \right) - \alpha_{k+1} \left(1 - \frac{\ln \alpha_{k+1}}{\alpha_{k+1}} \right) \\
&= \Bigg(\underset{\begin{matrix}
\downarrow\\
\infty
\end{matrix}}
{\underbrace{\frac{\alpha_{k+2}}{\alpha_{k+1}}}} \bigg(1 - \underset{\begin{matrix}
\downarrow\\
0
\end{matrix}}
{\underbrace{\frac{\ln \alpha_{k+2}}{\alpha_{k+2}}}} \bigg)- \bigg(1 - \underset{\begin{matrix}
\downarrow\\
0
\end{matrix}}
{\underbrace{\frac{\ln \alpha_{k+1}}{\alpha_{k+1}}}} \bigg)\Bigg)\alpha_{k+1} \ge \alpha_{k+1} \to \infty.
\end{align*}
Therefore $\theta_{k}\to\infty$ and, asymptotically, for $M$ large, $c_{k} \le c_{k+1}$. Since this holds for all three cases of $y_{k}$, we have that $\optflow{c}=\min\{c_{k},c_{k-1}\}$.

When $y_{k} > \alpha_{k+1}$, $\optflow{c}=c_{k}$ if
\begin{equation*}
c_{k}=\expon{M-\alpha_{k+1}}+\expon{\alpha_{k+1}} \le c_{k+1}=\expon{M-\alpha_{k}}+\expon{\alpha_{k}},
\end{equation*}
that is,
\begin{equation}\label{eq:expak+1ak}
\expon{\alpha_{k+1}}[1+\expon{M-2\alpha_{k+1}}] \le \expon{\alpha_{k}}[1+\expon{M-2\alpha_{k}}].
\end{equation}
Given that $\expon{M-2\alpha_{k+1}} \le 1$, the inequality in \eqref{eq:expak+1ak} is implied by
\begin{equation}\label{eq:2expak+1ak}
2\expon{\alpha_{k+1}} \le \expon{\alpha_{k}}[1+\expon{M-2\alpha_{k}}].
\end{equation}
Since $y_{k} > \alpha_{k+1}$, we have
\begin{equation*}
M \ge 2\alpha_{k+1}-\ln \alpha_{k+1} \Longrightarrow \expon{M-2\alpha_{k}} \ge \frac{\expon{2(\alpha_{k+1}-\alpha_{k})}}{\alpha_{k+1}}.
\end{equation*}
For this to hold, it suffices to have
\begin{equation*}
2\expon{2(\alpha_{k+1}-\alpha_{k})} \le 1+\frac{\expon{2(\alpha_{k+1}-\alpha_{k})}}{\alpha_{k+1}}.
\end{equation*}
We actually prove the stronger inequality
\begin{equation*}
2\expon{2(\alpha_{k+1}-\alpha_{k})} \le \frac{\expon{2(\alpha_{k+1}-\alpha_{k})}}{\alpha_{k+1}},
\end{equation*}
that is
\begin{equation*}
2\alpha_{k+1} \le \expon{\alpha_{k+1}-\alpha_{k}},
\end{equation*}
or
\begin{equation*}
\ln 2 + \ln \alpha_{k+1} \le \alpha_{k+1}-\alpha_{k},
\end{equation*}
which holds, since 
\begin{equation*}
\underset{
\begin{matrix}
\downarrow\\
0
\end{matrix}}
{\frac{\ln 2}{\alpha_{k+1}}} + \underset{
\begin{matrix}
\downarrow\\
0
\end{matrix}}
{\frac{\ln \alpha_{k+1}}{\alpha_{k+1}}} \le \underset{
\begin{matrix}
\downarrow\\
1
\end{matrix}}
{\frac{\alpha_{k+1}}{\alpha_{k+1}}} - \underset{
\begin{matrix}
\downarrow\\
0
\end{matrix}}
{\frac{\alpha_{k}}{\alpha_{k+1}}}.
\end{equation*}

When $y_{k}\le \alpha_{k}$, $\optflow{c}=c_{k-1}$ since
\begin{equation*}
c_{k}=\expon{M-\alpha_{k}}+\frac{\alpha_{k}}{\alpha_{k+1}}\expon{\alpha_{k+1}} \ge
\expon{M-\alpha_{k}}+\frac{\alpha_{k}}{\alpha_{k}}\expon{\alpha_{k}}=c_{k-1}.
\end{equation*}

When $\alpha_{k} < y_{k} \le \alpha_{k+1}$,  we have
\begin{align*}
c_{k}&=\frac{\expon{\alpha_{k+1}}}{\alpha_{k+1}}[1+M-\alpha_{k+1}\ln \alpha_{k+1}]\\
c_{k+1}&=\expon{M-\alpha_{k}}+\expon{\alpha_{k}}=\expon{-\alpha_{k}}[\expon{M}+\expon{2\alpha_{k}}].
\end{align*}
So $c_{k}\le c_{k+1}$ if
\begin{align*}
M\frac{\expon{\alpha_{k+1}}}{\alpha_{k+1}}-\frac{\expon{M}}{\expon{\alpha_{k}}} &\le \expon{\alpha_{k}} + \expon{\alpha_{k+1}}\left(1-\frac{1}{\alpha_{k+1}}-\frac{\ln \alpha_{k+1}}{\alpha_{k+1}} \right) \\
M\frac{\expon{\alpha_{k}+\alpha_{k+1}}}{\alpha_{k+1}}-\expon{M} &\le \expon{2\alpha_{k}} + \expon{\alpha_{k}+\alpha_{k+1}}\left(1-\frac{1}{\alpha_{k+1}}-\frac{\ln \alpha_{k+1}}{\alpha_{k+1}} \right).
\end{align*}
For $M$ slightly larger than $\alpha_{k}+\alpha_{k+1}$ the inequality becomes approximately
\begin{align*}
\frac{\alpha_{k}+\alpha_{k+1}}{\alpha_{k+1}}\expon{\alpha_{k}+\alpha_{k+1}}- \expon{\alpha_{k}+\alpha_{k+1}} &\le \expon{2\alpha_{k}} + \expon{\alpha_{k}+\alpha_{k+1}}\left(1-\frac{1}{\alpha_{k+1}}-\frac{\ln \alpha_{k+1}}{\alpha_{k+1}} \right) \\
\underset{
\begin{matrix}
\downarrow\\
0
\end{matrix}}{\frac{\alpha_{k}}{\alpha_{k+1}}} &\le \underset{
\begin{matrix}
\downarrow\\
0
\end{matrix}}{\expon{\alpha_{k}-\alpha_{k+1}}} + \bigg(1-\underset{
\begin{matrix}
\downarrow\\
0
\end{matrix}}{\frac{1}{\alpha_{k+1}}}-\underset{
\begin{matrix}
\downarrow\\
0
\end{matrix}}{\frac{\ln \alpha_{k+1}}{\alpha_{k+1}}} \bigg)
\end{align*}
Therefore for $M$ in a right neighborhood of $\alpha_{k}+\alpha_{k+1}$ we have $\optflow{c}=c_{k}$.

We can now compute the price of anarchy.
\begin{align*}
\PoA(\Gamma_{M})&=\frac{(\alpha_{k}+\alpha_{k+1})\frac{\expon{\alpha_{k+1}}}{\alpha_{k+1}}}{\frac{\expon{\alpha_{k+1}}}{\alpha_{k+1}}(1+\alpha_{k}+\alpha_{k+1}-\alpha_{k+1}+\ln \alpha_{k+1})}\\
&=\frac{\alpha_{k}+\alpha_{k+1}}{1+\alpha_{k}+\ln \alpha_{k+1}}\\
&=\frac{\alpha_{k+1}(1+\frac{\alpha_{k}}{\alpha_{k+1}})}{1+\alpha_{k}+\ln \alpha_{k+1}}\\
&=\frac{1+\frac{\alpha_{k}}{\alpha_{k+1}}}{\frac{1}{\alpha_{k+1}}+\frac{\alpha_{k}}{\alpha_{k+1}}+\frac{\ln \alpha_{k+1}}{\alpha_{k+1}}}\to\infty
\end{align*}
For $M$ in a right neighborhood of $\alpha_{k}+\alpha_{k+1}$.
Therefore
\begin{equation*}
\limsup_{M\to\infty}\PoA(\Gamma_{M})=+\infty. \qquad\qed
\end{equation*}
\end{proof}

\section{Figures}

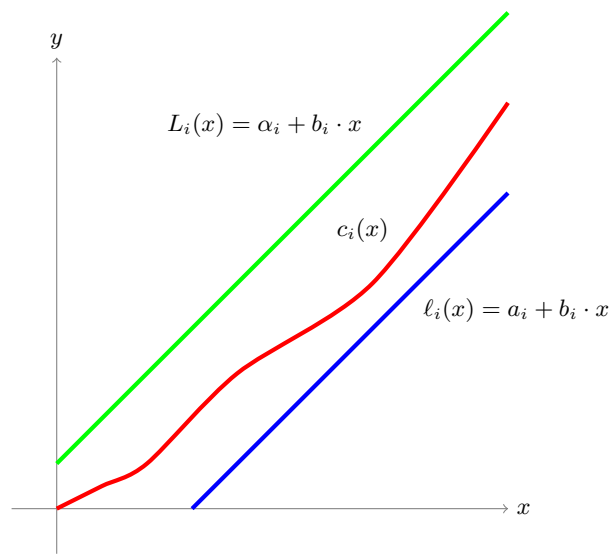
\begin{figure}[h]
\centering
{%
\begin{tikzpicture}[scale=.6]
    \draw [thin, gray, ->] (0,-1) -- (0,10)      
        node [above, black] {$y$};              

    \draw [thin, gray, ->] (-1,0) -- (10,0)      
        node [right, black] {$x$};              

    \draw [draw=green,ultra thick,mylabel=at 0.7 above left with {$L_{i}(x) = \alpha_{i}+b_{i}\cdot x$}] (0,1) -- (10,11);
    \draw [draw=blue,ultra thick,mylabel=at 0.7 below right with {$\ell_{i}(x) = a_{i} + b_{i}\cdot x$}] (3,0) -- (10,7);
    \draw [draw=red,ultra thick,mylabel=at 0.7 above left with {$c_{i}(x)$}] plot [smooth] coordinates {(0,0) (1,0.5) (2,1) (4,3) (7,5) (10,9)};
\end{tikzpicture}
}
~\vspace{0cm} \caption{\label{fi:affinebound} Affinely bounded costs.}
\end{figure}

\end{document}